\documentclass{llncs}
\usepackage[english]{babel}
\usepackage[utf8]{inputenc}
\usepackage[T1]{fontenc}
\usepackage{amsmath}
\usepackage{amsfonts}
\usepackage{amssymb}
\usepackage{dsfont} 
\usepackage[dvips]{graphics}
\usepackage{epsfig,psfrag}
\usepackage{stmaryrd}
\usepackage{color}

\begin{document}

\title{A chaos-based approach for information hiding security}
\author{Jacques M. Bahi, Senior Member IEEE\\Christophe Guyeux} 
\institute{Computer Science Laboratory LIFC\\ University of Franche-Comte, France}
\maketitle

\begin{abstract}
This paper introduces a new framework for data hiding security. Contrary to the existing ones, the approach introduced here is not based on probability theory. In this paper, a scheme is considered as secure if its behavior is proven unpredictable. The objective of this study is to enrich the existing notions of data hiding security with a new rigorous and practicable one. This new definition of security is based on the notion of topological chaos. It could be used to reinforce the confidence on a scheme previously proven as secure by other approaches and it could also be used to study some classes of attacks that currently cannot be studied by the existing security approaches. After presenting the theoretical framework of the study, a concrete example is detailed in order to show how our approach can be applied.
\end{abstract}

\emph{Keywords:} Information hiding security, Theory of chaos, Spread-spectrum.

\section{Introduction}

In past decades, the studies in the information hiding domain have almost exclusively been focused on robustness ~\cite{AdelsbachKS06},~\cite{BattiatoCGG99}. Security has emerged in the last years as a new interest in this domain~\cite{Cachin2004},~\cite{Mittelholzer99}, \cite{Perez-Freire2006:Security}. Security and robustness are neighboring concepts without clearly established definitions~\cite{Perez-Freire06}. Robustness is often considered to be mostly concerned by blind elementary attacks, whereas security is not limited to certain specific attacks. Security encompasses robustness and intentional attacks~\cite{ComesanaPP05bis},~\cite{Kalker2001}. The attempts to define the differences between robustness and security, to clarify the classes of attacks and to give some consistence to the notion of security, illustrates the recent important concern to bring a rigorous theoretical framework for security in data hiding.

In the framework of watermarking and steganography, security has seen several important developments since the last decade. Nevertheless, several open questions still remain. For example, the current theoretical approach requires strong hypotheses on the covertext and these hypotheses are quite difficult to evaluate in practice. Moreover, even if several security classes have been identified since the first classification of attacks, only a small part of them can be easily studied within this framework. In the existing approach, when a same hidden message is embedded into several covertexts, information leak is studied in terms of probability. If the leak is important, then the scheme is considered as insecure. In this paper we are interested in the evaluation of unpredictability of a data hiding scheme: this last is considered as secure if it is proven to be unpredictable. This new framework can be used to study some classes of attacks that are difficult to investigate in the existing security approach. It also enriches the variety of qualitative and quantitative tools that evaluate how strong the security is, thus reinforcing the confidence that can be put on a given scheme. Indeed, let us suppose that Eve, an attacker, observes the behavior of a data hiding machine. If there is no information leak, when Eve applies an input on the machine, and if nothing can be understanding when the input is changing, because of the unpredictable behavior of the machine, so she cannot deduce much things from these observations. This claim is discussed more rigorously in Section~\ref{SS}.

The rest of this paper is organized as follows. In Section~\ref{Refs}, related work concerning data hiding security is recalled. Our contribution in relation to existing framework is explained. In Section~\ref{section:Chaos}, the definition of chaos used in this paper is presented, the notion of chaotic-security is defined and the feasibility of its study in real-world applications is established. An illustrative example is given in Section~\ref{SS}: the chaotic-security study of various spread-spectrum data hiding schemes is presented and the level of security is qualitatively and quantitatively evaluated. The link between chaotic-security and stego-security is discussed and the impact of our study in regard to various classes of attacks is detailed. The paper ends with a conclusion where our contribution is summarized, and planned future work is discussed.

\section{Related work and contribution}
\label{Refs}
\subsection{Related work}
\label{raleted}

The first fundamental work in security was made by Cachin in the late 90's~\cite{Cachin2004} in the context of steganography. Cachin interprets the attempt of an attacker to distinguish between an innocent image and a stego-content as a hypothesis testing problem. In this document, the basic properties of a stegosystem are defined using the notions of entropy, mutual information, and relative entropy. Mittelholzer~\cite{Mittelholzer99}, inspired by the work of Cachin, proposed the first theoretical framework for analyzing the security of a watermarking scheme.

These efforts to bring a theoretical framework for security in steganography and watermarking, have been followed up by Kalker~\cite{Kalker2001} who tries to clarify the concepts (robustness \emph{vs.} security) and the classifications of watermarking attacks. This work has been deepened by Furon \emph{et al.}~\cite{Furon2002}, who have translated Kerckhoffs' principle (Alice and Bob shall only rely on some previously shared secret for privacy) from cryptography to data hiding. They used Diffie and Hellman methodology, and Shannon's cryptographic framework~\cite{Shannon49}.

These attacks have been classified into categories, according to the type of information Eve has access to~\cite{Cayre2005},~\cite{Perez06}:
\begin{itemize}
\item Watermarked Only Attack (WOA): the attacker has access only to watermarked contents.
\item Known Message Attack (KMA): the attacker has access to pairs of watermarked contents and corresponding hidden messages.
\item Known Original Attack (KOA): occurs when an attacker has access to several pairs of watermarked contents and their corresponding original versions.
\item Constant-Message Attack (CMA): the attacker observes several watermarked contents and only knows that the unknown hidden message is the same in all contents.
\end{itemize}

Each category of attack conducts to different security classes. For example, four classes of security are defined in \cite{Cayre2008} for WOA, namely insecurity, key-security, subspace-security and stego-security.

Barni \emph{et al.} proposed a different approach to watermarking security, based on games with some rules concerning the public available information~\cite{BarniBF03}, conducting to a definition of security level similar to that proposed by Furon. Lastly, Cayre \emph{et al.} proposed in~\cite{Cayre2005} the Fisher Information Matrix to quantify security in this context.

\subsection{Some reasons why unpredictability can improve security}

Stego-security is clearly relevant and required in WOA setup: Eve has only access to watermarked contents and due to stego-security, it is impossible for her to decide whether a content has been processed through the embedding function or not. So in WOA setup, a stego-secure algorithm can face Eve's attacks. However, such a framework is not as useful to counteract KOA, KMA and CMA classes of attacks. In these setups, Eve tries to take benefits of its observations of watermarked contents, when she changes some  initial conditions in the data hiding scheme. She desires having a sufficient understanding of the scheme and to be able to predict its behavior.

This knowledge can serve an attacker in various situations, for example when trying to counteract digital rights management (DRM), or in a man-in-the-middle attack through an hidden channel. Let us explain for example how Eve can try to achieve a man-in-the-middle attack by taking benefits of the predictable behavior of a scheme in KOA setup. We suppose that Alice and Bob communicate together through an hidden channel into some innocent master paintings. Eve has thus access to original and watermarked paintings, by using a base of knowledge and observing the communication channel. If she is able to predict the behavior of the data hiding scheme used by Alice and Bob, so when Alice send a watermarked painting $P_o$ to Bob: Eve intercepts $P_o$, use the same painting $P$ than Alice, and tries to predict how its own message should change $P$. The result of this prediction is sent to Bob. It is true that the chances for success for this attack are low, but these chances increase with the predictability of the data hiding scheme, thus revealing a security failure.

Let us now suppose that Eve wants to attack some DRM. She has access to several pairs of watermarks - the copyrights, which are public - and watermarked media: we are in KMA setup. In addition, she has access to the data hiding machine, but does not have knowledge of a secret key parameter, from which the outputs of the scheme are dependent. This key thus determines the way to apply the copyright into the media. She wants to insert its own copyright in this protected media, to make it impossible to determine whether Eve is the owner or not. She does not know exactly how copyrights are applied on original media, because the DRM machine works with a secret key. However Eve can reach its goal if she is able to predict the behavior of the copyright machine: she can approximately determine what should be the watermarked media with its own copyright. As a conclusion, in various situations, unpredictable schemes should be recommended to improve security in data hiding. 

\subsection{Contributions in this paper}
\label{contrib}
In this paper, a novel theoretical framework for data hiding security is proposed. A data hiding scheme is considered here as a machine, whose detail is public. This machine receives hidden messages and original contents from the outside world and returns stego-contents. In our point of view, security of the scheme depends on the unpredictable behavior of the machine. To give consistence to the notion of unpredictability, this machine is modeled as a dynamical system: $x^{n+1}=f(x^{n})$ ~(where $x^n$ denotes the $n^{th}$ term of the sequence $x$). This reformulation is always possible, as it is proven in Section~\ref{CS}.

Thus, unpredictability refers to some topological or ergodic aspects of $f$ taken from the mathematical theory of chaos, as defined by Devaney \cite{Devaney}, Li-Yorke~\cite{Li75}, or Adler~\cite{Adler65} for example. This new theoretical framework for security respects Kerckhoffs'principle. It is based on a topological description of data hiding, whereas most studies in this field usually have used the theory of probability~\cite{Perez-Freire06},~\cite{Furon05}. The goal of this research work is to give additional contribution to the variety of security evaluations which should lead to better confidence into data hiding schemes.

Compared to the information-theoretic model for steganography proposed by Cachin \cite{Cachin2004} and extended for example by Ker~\cite{Ker06}, chaotic-security appears to be a little more relevant and probably more realizable than stego-security in some particular situations, for example when chaotic sequences are used in data hiding schemes. This framework can check whether the claim of a chaotic behavior for a data hiding scheme can be verified or not. Indeed, to our best knowledge, stego-security studies often take place in WOA category and are related to Simmons' prisoners' problem~\cite{Simmons83}. In this problem, Alice and Bob are in jail and want to plan an escape by exchanging hidden messages in innocent-looking cover contents. These messages are conveyed by Eve, a warden who tries to benefit from any information leak resulting by the use of the same secret key. Quoting Cayre and Bas in~\cite{Cayre2008}: ``Like other works, we consider Alice and Bob use only one key. Of course, in real applications, especially in steganography, it is highly desirable to change the key at every communication between Alice and Bob.'' In addition, a probabilistic model of the covertext is needed and, as stated by Cachin in \cite{Cachin2004}, ``assuming the existence of a covertext distribution seems to render our model somewhat unrealistic for the practical purpose of steganography''. The new framework we propose does not suppose any assumption of this kind, works with simple or multiple secret keys and is not restricted to some category of attacks. Our approach is explained in detail in the following section.


\section{Chaos for data hiding security}
\label{section:Chaos}

We will consider that a data hiding scheme is secure when its behavior is unpredictable. The mathematical domain that studies unpredictability is the theory of chaos, which describes the behavior of a dynamical system in topological terms. One of the most reputed description of such a chaotic behavior is due to Devaney~\cite{Devaney}. It is recalled in the next subsection, whereas the notion of chaotic-security deduced from this definition is presented in Subsection~\ref{CS}.

\subsection{Devaney's chaotic dynamical systems}
\label{subsection:Devaney}

Consider a metric space $(\mathcal{X},d)$ and a continuous function $f$ on $\mathcal{X}$. Let $f^{k}=f\circ ...\circ f$ denotes the $k^{th}$ composition of a function $f$. Quoting Devaney in~\cite{Devaney},

\begin{definition}
$f$ is said to be \emph{topologically transitive} if, for any pair of open sets $U,V \subset \mathcal{X}$, there exists $k>0$ such that $f^k(U) \cap V \neq \varnothing$.
\end{definition}

\begin{definition}
An element (a point) $x$ is a \emph{periodic element} (point) for $f$ of period $n\in \mathds{N}^*,$ if $f^{n}(x)=x$. The set of periodic points of $f$ is denoted $Per(f).$
\end{definition}

\begin{definition}
$\{\mathcal{X},f\}$ is said to be \emph{regular} if the set of periodic points is dense in $\mathcal{X}$,
\begin{equation*}
\forall x\in \mathcal{X},\forall \varepsilon >0,\exists p\in Per(f),d(x,p)\leqslant \varepsilon .
\end{equation*}
\end{definition}

\begin{definition}
\label{sensitivity} $f$ has \emph{sensitive dependence on initial conditions}
if there exists $\delta >0$ such that, for any $x\in \mathcal{X}$ and any neighborhood $V$ of $x$, there exists $y\in V$ and $n\geqslant 0$ such that $|f^{n}(x)-f^{n}(y)|>\delta $.

$\delta$ is called the \emph{constant of sensitivity} of $f$.
\end{definition}

\begin{definition}
A function $f:\mathcal{X}\longrightarrow \mathcal{X}$ is said to be \emph{chaotic} on $\mathcal{X}$ if $\{\mathcal{X},f\}$ is regular, topologically transitive and has sensitive dependence on initial conditions.
\end{definition}

When $f$ is chaotic, then the system $\{\mathcal{X}, f\}$ is chaotic and quoting Devaney~\cite{Devaney}: ``it is unpredictable because of the sensitive dependence on initial conditions. It cannot be bren down or simplified into two subsystems which do not interact because of topological transitivity. And in the midst of this random behavior, we nevertheless have an element of regularity''. Fundamentally different behaviors are consequently possible and occurs in an unpredictable way.

\subsection{Chaotic-security}
\label{CS}

As stated before, we believe that an unpredictable behavior is required for a data hiding scheme to satisfy an efficient level of security. This unpredictability makes it difficult to determine whose coefficients of the cover media will be altered during the embedding of the watermark, which limits the possibilities of Eve in KOA and KMA attacks. The scheme should at least be chaotic according to Devaney: this property will improve the ability of Alice and Bob to withstand attacks. Indeed, it will be as difficult for Eve to find the hidden message after $n$ iterations than to forecast the weather after $n$ days without mistakes: because of chaos, this last becomes impossible to do in practice when $n$ increases. Due to Devaney's chaos, such a ``chaotic-secure'' information hiding scheme will thus satisfy sensitive dependence to the initial condition, regularity and transitivity. Sensitivity to initial conditions is useful, among other, to withstand sensitivity attack~\cite{ComesanaPP05} (it can be noticed that, in the context of the sensitivity attack, various strategies have been already employed~\cite{Furon08}). In addition, sensitivity is helpful to achieve authentication, because the watermark's embedding will be highly dependent on any changes of the carrier image. Moreover, fragile data hiding is achieved with a large constant of sensitivity. 
Transitivity trends to improves robustness: for example, Eve cannot hope to remove the watermark by cropping the media. Indeed, the system will visit all the space, so the watermark will be uniformly distributed on the whole media. This property improves authentication, as this last can be achieved by studying any part of the media: theoretically speaking, authentication still remains possible in a cropped media. Transitivity trends to improve security too, as Eve cannot have a better understanding of the scheme, cannot reduce its complexity by studying only a well chosen reduced part of the watermarked content. 
Lastly, transitivity and regularity lead to unpredictability, which helps Alice and Bob to withstands KOA and KMA attacks.
For these reasons we believe that this new point of view could enrich the field of security in data hiding.

\medskip

Let us now present more rigorously the new notion of chaotic-security. To check whether an existing data hiding scheme is chaotic or not, we propose firstly to write it as an iterate process $x^{n+1}=f(x^n)$. It is possible to prove that this formulation can always be done. Let us consider a given data hiding algorithm. It is always possible to translate it as a Turing machine and this last can be written as $x^{n+1} = f(x^n)$ by the following way. Let $(w,i,q)$ be the current configuration of the Turing machine (Fig.~\ref{Turing}), where $w=\sharp^{-\omega} w(0) \hdots w(k)\sharp^{\omega}$ is the paper tape, $i$ is the position of the tape head, $q$ is used for the state of the machine, and $\delta$ is its transition function. We define $f$ by:
\begin{itemize}
\item $f(w(0) \hdots w(k),i,q) = ( w(0) \hdots w(i-1)aw(i+1)w(k),i+1,q')$, \newline if  $\delta(q,w(i)) = (q',a,\rightarrow)$,
\item $f( w(0) \hdots w(k),i,q) = (w(0) \hdots w(i-1)aw(i+1)w(k),i-1,q')$, \newline if $\delta(q,w(i)) = (q',a,\leftarrow)$.
\end{itemize}
Thus the Turing machine can be written as an iterate function $x^{n+1}=f(x^n)$ on a well-defined set $\mathcal{X}$, with $x^0$ as the initial configuration of the machine. We denote by $\mathcal{T}(S)$ the iterative process of a data hiding scheme $S$.

\begin{figure}[h]
  \centering
\includegraphics[scale=0.5]{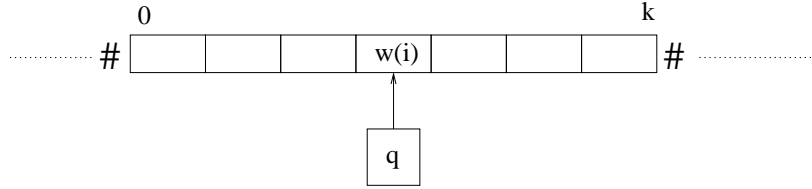}
\caption{Turing Machine}
\label{Turing}
\end{figure}

Let $d$ be a relevant distance on $\mathcal{X}$. So the behavior of this dynamical system can be studied to know whether the data hiding scheme is unpredictable or not. This leads to the following definition.

\begin{definition}
An information hiding scheme $S$ is said to be chaotic-secure on $(\mathcal{X},d)$ if its iterative process $\mathcal{T}(S)$ has a chaotic behavior according to Devaney.
\end{definition}

Theoretically speaking, chaotic security can always be studied, as it only requires that the two following points are satisfied.
\begin{itemize}
\item Firstly, the data hiding scheme must be written as an iterate function on a set $\mathcal{X}$. We have stated that this is possible for any given scheme.
\item Secondly, a metric or a topology must be defined on $\mathcal{X}$. This is always possible, for example by taking either the discrete or the trivial topology, even though these last are not really relevant for the aims we want to reach.
\end{itemize}

Chaotic-security is clearly impacted by the choice of the distance or topology on $\mathcal{X}$ and this dependence must be regarded with attention. It is evident that the choice of the metric (or of the topology) must be justified, for example by establishing a strong link between the proximity of two points and the aims that data hiding attempts to reach.
However, some topologies are more natural and reasonable than others, and equivalence of topologies reduces the impact of this choice. In addition, it can be remarked that stego-security supposes the same kind of hypothesis: dealing with probabilities implies the definition of a sigma-algebra. To our best knowledge, the Borel algebra is always chosen, even though this choice is neither stated, nor justified. As a topological space is needed to define Borel sets, we can claim that at least when a stego-security study is possible, so a chaotic study can be realized, with a topology inherited from the stego-security study (which justifies its choice). 

\medskip

In our point of view, the chaotic-security is the lowest level of security for a data hiding scheme in terms of unpredictability. This property is required, but is not sufficient: it is only the first stage of the evaluation of the unpredictable behavior of the scheme. This study must be followed by the establishment of the list of chaotic properties that the system presents. Indeed, being unpredictable is a tricky thing to define and the number of candidates that give consistence to this notion is large. Namely topological and metric entropy, ergodicity, topological mixing, lyapunov exponent, expansivity, transitivity and strong transitivity, bifurcation theory, or chaos as defined by Li-Yorke, by Devaney or by Knudsen to name a few. As each definition illustrates a particular aspect of an unpredictable behavior and has its own interest, each notion of chaos offers a new light on the security of a data hiding scheme. Thus we consider that a given data hiding scheme will be more secure than another one if it presents a larger number of chaotic qualities and if its quantitative values are better. Indeed the properties to check is depending on the aims to reach: fragile watermarking, robustness, \emph{etc.} This point is illustrated in the next section and will be largely deepened in a future work.

\section{Chaotic-security of spread-spectrum data hiding schemes}
\label{SS}

\subsection{A first proof of chaotic-security}
\label{spread}


In what follows, our framework is used to give a first chaotic-security evaluation of the well-known spread-spectrum (SS) data hiding techniques~\cite{Cayre2008}. This proves that the previous framework is ready for real-world applications and establishes a concrete link between chaotic and stego-security notions.

\medskip

Let $x \in \mathds{R}^{N_v}$ be an host vector in which we want to hide a message $m\in \{0;1\}^{N_c}$. $N_c$ is the size of the hidden payload (in bits) and $N_v$ the size of the stego or host vector (in samples). A key $\mathcal{K}$ is used to initialize a PRNG (Pseudo-Random Number Generator) to obtain $N_c$ secret carries $\{u^i\}$ taken in $\mathds{R}^{N_v}$. Thus in classical SS the watermark signal $w$ is constructed as follows:
\begin{equation}
\displaystyle{w = \sum_{i=0}^{N_c-1} \gamma (-1)^{m^i} u^i}
\end{equation}

\noindent where $\gamma$ is a given distortion level. The watermarked signal $y$ is then defined by:
\begin{equation}
y = x + w
\end{equation}

Let us now suppose that the components of the watermark are bounded by a finite value $\mathsf{N}$: max$\left(\{ w_i, i \in \llbracket 1, N_v \rrbracket \}\right) \leqslant \mathsf{N}$. This bound can be as large as needed, however a very large $\mathsf{N}$ seems to be contradictory with the aims of a data hiding scheme.
Let us consider $\mathcal{X}=\left(\left[0,\mathsf{N}\right]^{N_v}\right)^\mathds{N}\times \mathds{R}^{N_v}$ and
\begin{equation}
G((S,E)) = (\sigma(S) ; i(S) + E)
\end{equation}

\noindent where $\sigma$ is the \emph{shift} function defined by $\sigma :(S^{n})_{n\in \mathds{N}}\in \left(\left[0,\mathsf{N}\right]^{N_v}\right)^\mathds{N} \rightarrow (S^{n+1})_{n\in \mathds{N}}\in  \left(\left[0,\mathsf{N}\right]^{N_v}\right)^\mathds{N} $ and the \emph{initial function} $i$ is the map which associates to a sequence, its first term: $i:(S^{n})_{n\in \mathds{N}}\in  \left(\left[0,\mathsf{N}\right]^{N_v}\right)^\mathds{N}\rightarrow S^{0}\in  [0;\mathsf{N}]^{N_v}$.

Spread-spectrum data hiding techniques are thus the result of $N_c$ iterations of the following dynamical system:
\begin{equation}
\left\{
\begin{array}{l}
X^0 \in \mathcal{X},\\
X^{n+1} = G(X^n),
\end{array}
\right.
\end{equation}

\noindent and the watermarked media is the second component of $X^{N_c}$.

Classical SS, \emph{i.e.} with BPSK modulation~\cite{Cayre2008}, is defined by $X^0 = (S^0, E^0)$ where $E^0$ is the host vector $x$ and $S^0$ is the sequence 
\begin{equation}
\left((-1)^{m^0} \gamma ~u^0, (-1)^{m^1} \gamma ~u^1, \hdots, (-1)^{m^{N_c-1}} \gamma ~u^{N_c-1} \right),
\end{equation}

\noindent in which $\gamma$ allows to achieve a given distortion, whereas in ISS (Improved Spread Spectrum~\cite{Malvar03}), $S^0$ is defined by
\begin{equation}
\left( (-1)^{m^i} \alpha  - \dfrac{<x,u^i>}{||u^i||^2} \right)_{i=0, \hdots, N_c-1},
\end{equation}

\noindent where $\alpha$ and $\lambda$ are computed to achieve an average distortion and to minimize the error probability~\cite{Cayre2008}. Lastly, in natural watermarking NW, $S^0$ is defined by
\begin{equation}
\left(- \left( 1 + \eta (-1)^{m^i} \dfrac{<x,u^i>}{|<x,u^i>|} \right) \dfrac{<x,u^i>}{||u^i||^2}
\right)_{i=0, \hdots, N_c-1}.
\end{equation}

This modulation consists in a model-based projection on the different vectors $u^i$ followed by a $\eta-$scaling along the direction of $u^i$. Natural watermarking has been proven stego-secure in~\cite{Cayre2008}.

\bigskip

We will prove in what follows that spread-spectrum data hiding schemes are chaotic-secure, \emph{i.e.} that $G$ is chaotic on $(\mathcal{X},d)$, thus finding a first chaotic and stego-secure algorithm.

Let $d_{\infty}(A,B) = max\left\{ |A_i-B_i|, i=1\hdots N_v\right\}$ be one of the usual metrics on $\mathds{R}^{N_v}$. We define a new distance between two points $X = (S,E), Y = (\check{S},\check{E})\in
\mathcal{X}$ by $d(X,Y)=d_{\infty}(E,\check{E})+d_s(S,\check{S}),$ where $\displaystyle{d_s(S,\check{S})} = \displaystyle{\dfrac{9}{\mathsf{N}}%
\sum_{k=0}^{\infty }\dfrac{d_{\infty}(S^k, \check{S}^k)}{10^{k}}}.$

The choice of $d_{\infty}$ on $\mathds{R}^{N_v}$ is not important, because of the equivalence of norms in finite dimension: as topologies are the same, thus chaotic properties does not change by using another distance on $\mathds{R}^{N_v}$. $d_s$ has been chosen such that $d(X,Y)$ is small when the distance between the watermarked images resulting on the spread-spectrum applied on $X$ and $Y$ are close (for any metrics on $\mathds{R}^{N_v}$, as they are all equivalent). Lastly, $\dfrac{9}{\mathsf{N}}$ is just a normalization value.

We will now prove that:

\begin{proposition}
$G$ is continuous on ($\mathcal{X},d$).
\end{proposition}

\begin{proof}
We use the sequential continuity. Let $(S_n,E_n)_{n\in \mathds{N}}$ be a sequence of the phase space $\mathcal{X}$, which converges to $(S,E)$. We will prove that $\left( G(S_n,E_n)\right)_{n\in \mathds{N}}$ converges to $G(S,E)$. Let us recall that for all $n$, $S_n$ is a strategy, thus, we consider a sequence of strategies (\emph{i.e.} a sequence of sequences).

As $d((S_n,E_n);(S,E))$ converges to 0, each distance $d_{\infty}(E_n,E)$ and $d_s(S_n,S)$ converges to 0. 

\begin{enumerate}
\item If $\displaystyle{\dfrac{9}{\mathsf{N}}%
\sum_{k=0}^{\infty }\dfrac{d_{\infty}(S_n^k, S^k)}{10^{k}} \rightarrow 0}$ when $n \rightarrow \infty$, then $\displaystyle{\dfrac{9}{\mathsf{N}}\sum_{k=1}^{\infty }\dfrac{d_{\infty}(S_n^k, S^k)}{10^{k}} \rightarrow 0}$. So $\displaystyle{\dfrac{9}{\mathsf{N}}\sum_{k=0}^{\infty }\dfrac{d_{\infty}(S_n^{k+1}, S^{k+1})}{10^{k+1}}}$ $=\dfrac{1}{10} d_s(\sigma(S_n);\sigma(S)) \rightarrow 0$. As a consequence,\linebreak $d_s(\sigma(S_n),\sigma(S))$ converges to 0.

\item Let us prove that $d_{\infty}\left(i(S_n)+E_n;i(S)+E\right) \rightarrow 0$. 
$$
\begin{array}{ll}
d_{\infty}\left(i(S_n)+E_n;i(S)+E\right) & = max\left\{ \left|\left(i(S_n)_k+(E_n)_k\right) - \left( i(S)_k + E_k \right) \right|, k=1\hdots N_v\right\} \\
& = max\left\{ \left|\left(i(S_n)_k - i(S)_k \right) + \left( (E_n)_k - E_k \right) \right|, k=1\hdots N_v\right\}\\
& \leqslant  max\left\{ \left| i(S_n)_k - i(S)_k  \right|, k=1\hdots N_v\right\}$ + $d_{\infty} (E_n,E) \\
& = d_\infty ( S_n^0, S^0) + d_{\infty}(E_n,E) \\
&  \leqslant d_s (S_n, S) + d_{\infty} (E_n, E) \\
&  = d \left( (S_n,E_n) ; (S,E) \right) \rightarrow 0. 
\end{array}$$
\end{enumerate}
\end{proof}

\begin{proposition}
Periodic points of $G$ are dense in $\mathcal{X}$, so $G$ is regular.
\end{proposition}

\begin{proof}
Let $(S,E) \in \mathcal{X}$ and $\varepsilon > 0$. We are looking for a periodic point $(\check{S}, \check{E}) \in \mathcal{X}$ such that $d\left((S,E), (\check{S},\check{E})\right) < \varepsilon$. Let $\check{E} = E$ and $S_n$ denotes the sequence defined by:
$$\left\{ 
\begin{array}{ll}
S_n^k = S^k & \forall k \leqslant n \\
S_n^k = (\mathsf{N}, \hdots, \mathsf{N}) & \textrm{if } k>n \textrm{ and } k \equiv 0 \textrm{ (mod } 2 \textrm{)} \\
S_n^k = (-\mathsf{N}, \hdots, -\mathsf{N}) & \textrm{else.}
\end{array}
\right.$$
Then $d_s(S_n, S) = \displaystyle{\dfrac{9}{\mathsf{N}}\sum_{k=n+1}^{\infty }\dfrac{d_{\infty}(S_n^k, S^k)}{10^{k}}} \leqslant \displaystyle{\dfrac{9}{\mathsf{N}}\sum_{k=n+1}^{\infty }\dfrac{\mathsf{N}}{10^k}} = \dfrac{1}{10^n} \rightarrow 0$ when \linebreak $n \rightarrow \infty$. So $\exists n_0 \in \mathds{N}$ such that $d_s(S^{n_0}, S) < \varepsilon$. The point $(S^{n_0}, E)$ is then a periodic point of $\mathcal{X}$ which is $\varepsilon$-close to the given point $(S,E)$.
\end{proof}

We will now prove that,

\begin{proposition}
$G$ is transitive on $\mathcal{X}$.
\end{proposition}

\begin{proof}
Let $B_A = \mathcal{B}\left(X_A,r_A\right)$ and $B_B = \mathcal{B}\left(X_B,r_B\right)$ be two open balls of $\mathcal{X}$, where $X_A=(S_A,E_A)$ and $X_B=(S_B,E_B)$. We are looking for $\tilde{X}=(\tilde{S},\tilde{E}) \in B_A$ such that $\exists n_0 \in \mathds{N}, G^{n_0} ( \tilde{X} ) \in B_B$.

Let $k_0 \in \mathds{Z}$ such that $10^{-k_0}\leqslant r_A < 10^{-k_0+1}$ and $\left(\check{S}, \check{E}\right) = G^{k_0} \left( X_A \right)$. We define $\tilde{X} = (\tilde{S},\tilde{E})$ as below:
\begin{itemize}
\item $\tilde{E} = E_A$,
\item $\forall k \leqslant k_0, \tilde{S}^k = S_A^k$,
\item $\forall k \in \llbracket 1, N_v \rrbracket, \tilde{S}^{k_0+k} = (-\check{E}^k+E_B^k) \times (0, \hdots, 0, 1, 0,\hdots, 0)$, \emph{i.e.} the vector $\tilde{S}^{k_0+k}$ has its components null, except its $k^{th}$, equals to $(-\check{E}^k+E_B^k)$,
\item $\forall k \in \mathds{N}, \tilde{S}^{k_0+N_v+k+1} = S_B^k$.
\end{itemize}
With such a definition, $\tilde{X}$ is in $B_A$ and satisfies $G^{k_0+N_v} \left(\tilde{X} \right) \in B_B$.
\end{proof}

As $G$ is regular and transitive on $(\mathcal{X},d)$, we can conclude that $G$ is sensitive to initial conditions (using the result of Banks~\cite{Banks92}), thus proving that $G$ is chaotic in the meaning of Devaney. As a conclusion,

\begin{theorem}
Spread-spectrum data hiding techniques are chaotic-secure.
\end{theorem}

All the variety of SS techniques are concerned by this property of chaotic-security. In the point of view presented above, the choice of NW instead of ISS only affects the initial condition of the iterations of $G$. Indeed, the theory of chaos gives a global approach of the unpredictable behavior of a given system, but does not explain how to choose a good initial condition. For example, the reputed logistic map $X^0 \in [0,1], X^{n+1}=4 X^n (1-X^n)$, has a chaotic behavior, but if we choose $X^0=0$, then $\forall n \in \mathds{N}, X^n = 0$... We believe that stego-security is helpful to determine the initial values: to increase the security level of a given scheme, we thus recommend to use a chaotic-secure algorithm with stego-secure initial conditions. This discussion will be deepened in details in a future work.

\subsection{Qualitative and quantitative evaluation}

As stated before, the proof that a given data hiding scheme is chaotic-secure, is just the beginning of the study. The next stage is to evaluate the quality of its chaotic behavior, by using the numerous qualitative and quantitative tools offered by the theory of chaos. These tools allow to compare two given chaotic-secure schemes, by deciding which scheme is the most unpredictable and thus must be preferred. To give illustration, some tools are recalled in this section, namely strong transitivity, and the constants of expansivity and sensitivity. We will use them to give a better understanding of the unpredictability of spread-spectrum techniques. 

\subsubsection{Qualitative property: Strong transitivity}

\begin{definition}
A discrete dynamical system $\{\mathcal{X}, f\}$ is said to be \emph{strongly transitive} if and only if, for any point $A,B \in \mathcal{X}$ and any neighborhood $V$ of $B$, $n_0 \in \mathbb{N}$ and $X\in V$ can be found such that $f^n(X) = A$.
\end{definition}

We have the result,

\begin{proposition}
$\{\mathcal{X},G\}$ is strongly transitive.
\end{proposition}

\begin{proof}
Let us reconsider the proof of the transitivity of $(\mathcal{X}, G)$: we have defined $\tilde{X} \in B_A$ such that $G^{k_0+N_v} \left(\tilde{X}\right) \in B_B$. Indeed, for this $\tilde{X}$, we have: $G^{k_0+N_v} \left(\tilde{X}\right) = X_B$.
\end{proof}

\subsubsection{Quantitative measures}
\label{QUANTITATIVE MEASURE}

\label{par:Sensitivity}
One of the most famous measure in the theory of chaos is the constant of sensitivity defined in Definition \ref{sensitivity}. Intuitively, a function $f$ having a constant  sensitivity equal to $\delta $ implies that there exists points arbitrarily close to any point $x$ which \emph{eventually} separate from $x$ by at least $\delta $ under some iterations of $f$. This induces that an arbitrarily small error on an initial condition may magnified upon iterations of $f$. The sensitive dependence on the initial conditions is a consequence of regularity and transitivity~\cite{Banks92}. However, the constant of sensitivity can be obtained by proving the sensitivity without Banks' theorem.

\begin{proposition}
Spread-spectrum data hiding techniques $\{\mathcal{X},G\}$ have sensitive dependence on initial conditions and its constant of sensitivity is equal to $\dfrac{\mathsf{N}}{2}$.
\end{proposition}

\begin{proof}
Let $X = (S,E) \in \mathcal{X}$, $B = \mathcal{B}\left(X,r\right)$ an open ball centered in $X$, and $k_0 \in \mathds{Z}$ such that $10^{-k_0} \leqslant r < 10^{-k_0+1}$.
We define $\check{X}$ by:
\begin{itemize}
\item $\check{E} = E$,
\item $\check{S}^k = S^k$, $\forall k \in \mathds{N}$ such that $k \neq k_0+1$,
\item if $S_1^{k_0+1} < \frac{\mathsf{N}}{2}$, then $\check{S}_1^{k_0+1} = \mathsf{N}$, else $\check{S}_1^{k_0+1} = 0$,
\item $\forall i \in \llbracket 2, N_v \rrbracket, \check{S}_i^{k_0+1} = S_i^{k_0+1}$.
\end{itemize} 

So $d(X, \check{X}) = D_\infty (E,\check{E})+d_S(S,\check{S}) $ $= 0 + \dfrac{9}{\mathsf{N}} \dfrac{d_\infty (S^{k_0+1}, \check{S}^{k_0+1})}{10^{k_0+1}}$
$ \leqslant \dfrac{9}{\mathsf{N}} \dfrac{\mathsf{N}}{10^{k_0+1}} \leqslant \dfrac{1}{10^{k_0}} \leqslant r$, then $\check{X} \in B$. Let us now define $\mathcal{E} : \mathcal{X} \rightarrow \mathcal{X}, (S,E) \mapsto E$. So $\mathcal{E}\left(G^{k_0+1} (X) \right)_0 = \mathcal{E}\left(G^{k_0+1} (\check{X}) \right)_0$, because $E=\check{E}$ and $S^k = \check{S}^k, \forall k \leqslant k_0+1$. As:
\begin{itemize}
\item $\mathcal{E}\left(G^{k_0+2} (X) \right)_0 = \mathcal{E}\left(G^{k_0+1} (X) \right)_0 + S_0^{k_0+1}$,
\item $\mathcal{E}\left(G^{k_0+2} (\check{X}) \right)_0 = \mathcal{E}\left(G^{k_0+1} (\check{X}) \right)_0 + \check{S}_0^{k_0+1}$,
\item $\left| S_0^{k_0+1} - \check{S}_0^{k_0+1} \right| \geqslant \dfrac{\mathsf{N}}{2}$.
\end{itemize}

We thus have $d\left( G^{k_0+2}(X), G^{k_0+2}(\check{X}) \right)$ $\geqslant d_\infty \left(\mathcal{E}\left( G^{k_0+2}(X)\right), \mathcal{E}\left(G^{k_0+2}(\check{X})\right) \right)$
$\geqslant \left| \mathcal{E}\left( G^{k_0+2}(X)\right)_0 -  \mathcal{E}\left(G^{k_0+2}(\check{X})\right)_0\right| \geqslant \dfrac{\mathsf{N}}{2}$.
\end{proof}

Let us now recalled another usual quantitative measure of disorder.

\begin{definition}
A function $f$ is said to have the property of \emph{expansivity} if
\begin{equation*}
\exists \varepsilon >0,\forall x\neq y,\exists n\in \mathbb{N}%
,d(f^{n}(x),f^{n}(y))\geqslant \varepsilon .
\end{equation*}
\end{definition}

Then $\varepsilon $ is the \emph{constant of expansivity} of $f$: an arbitrarily small error on any initial condition is always amplified till $\varepsilon $.

\begin{proposition}
$\{\mathcal{X},G\}$ is not an expansive chaotic system.
\end{proposition}

\begin{proof}
Let $\varepsilon > 0$. We define: $X = \left( O_{N_v} ; (O_{N_v}, O_{N_v}, \hdots, O_{N_v}, \hdots ) \right)$ and $Y = \left( O_{N_v} ; (\dfrac{\varepsilon}{2} I_{N_v}, -\dfrac{\varepsilon}{2} I_{N_v}, \hdots, \dfrac{(-1)^n \varepsilon}{2} I_{N_v}, \hdots )\right)$, where $O_{N_v} = (0, \hdots, 0)$ is the null vector of size $N_v$ and $I_{N_v}$ is the vector of size $N_v$ equal to $(1,0, \hdots, 0)$. Thus, for this two points, we have:
$\forall n \in \mathds{N}, d\left(G^n (X) ; G^n (Y) \right) \leqslant \varepsilon .$
\end{proof}

\subsection{Discussion}
\label{Discussion}

Let us give now some consequences of this qualitative and quantitative evaluation. First of all, we can regret that spread-spectrum is not expansive. This property reduces drastically the benefits that Eve can obtain of an attack in KMA or KOA setup. For example, it is impossible to have an estimation of the watermark by moving the message (or the cover) as a cursor in situation of expansivity: this cursor will be too much sensible and the changes will be too much important to be useful. 
On the contrary, a very large constant of expansivity $\varepsilon$ is unsuitable: the cover media will be strongly altered whereas the watermark should be undetectable. Indeed, let us consider the same cover twice with two different watermarks. Thus $d(X,Y) < 1$ for the distance defined previously. However, due to expansivity, $\exists n \in \mathds{N}, d\left(G^n (X) ; G^n (Y) \right) \geqslant \varepsilon.$ Thus, $d_\infty \left(G^n (X)_1 ; G^n (Y)_1 \right) \geqslant \varepsilon - 1$, so either $d_\infty \left(X_1 ; G^n (X)_1 \right) \geqslant \dfrac{\varepsilon - 1}{2}$, or $d_\infty \left(Y_1 ; G^n (Y)_1 \right) \geqslant \dfrac{\varepsilon - 1}{2}$. If $\varepsilon$ is large, then at least one of the two watermarked media will be very different than its original cover.
Due to strong transitivity, the set of watermarked media obtained when using a fixed watermark, is equal to the whole set of media. In that situation, Eve cannot reduce the set of media to studied, reducing so the interest of a Constant-Message Attack setup for Eve. The importance of the sensitivity has been explained previously.

We will now discuss various consequences of the introduction of this new framework for security. Firstly, new comparisons can be done between given data hiding schemes, as it is illustrated by the following example. In a previous work~\cite{guyeux10}, we have proven that a famous tool in discrete dynamical systems, namely chaotic iterations, satisfies the Devaney's definition of chaos and we have proposed various applications of this tool in information security~\cite{internet09}, ~\cite{guyeux09}, ~\cite{guyeux10}. These chaotic iterations can be used to realize a data hiding scheme, as it is explained in~\cite{arxiv}. Qualitative and quantitative chaotic tools introduced in Section~\ref{SS} can thus be used to compare this algorithm to the spread-spectrum data hiding schemes. It can be proven that data hiding schemes based on chaotic iterations are chaotic secure, with the additional property of strong transitivity, as it is the case of spread-spectrum. However, chaotic iteration data hiding schemes have a larger constant of sensitivity than spread-spectrum and contrary to this last, chaotic iterations are expansive (with a constant of expansivity equal to 1). Moreover, chaotic iterations are topologically mixing, proving so that chaotic iterations appear to be more suitable than spread-spectrum to withstand attacks in KOA, KMA and CMA setup. All these claims will be proven in a future work.

\section{Conclusion and future work}

In this paper, a new concept of security for data hiding schemes is proposed as a complementary approach to the existing framework. This new notion of security contributes to the reinforcement of confidence into existing secure data hiding schemes. Moreover, it can replace stego-security in situations that are not encompassed by it. In particular, this framework is more relevant to give evaluation of data hiding schemes based on chaotic maps. 

In our approach, a data hiding scheme is secure if it is unpredictable. Its iterative process must satisfy the Devaney's chaos property and its level of security increases with the number of chaotic properties satisfied by it. This point has been clarified in Section~\ref{SS}, in which a first study of chaotic-security is proposed using some qualitative and quantitative tools taken from the mathematical theory of chaos.

We have shown in this paper that the intersection between the two sets of stego-secure and  chaotic-secure data hiding algorithms is nonempty, due to spread-spectrum techniques. This establishes a first connection between the two approaches for security in data hiding. In future work, we will discuss with more attention this fact, to give a better understanding of the links between these two frameworks. The comparison between spread-spectrum and chaotic iterations outlined in Section~\ref{Discussion} will be deepened. 
In addition, new tools taken from the theory of chaos will be introduced to enrich chaotic-security. Moreover, these tools will be compared to the Fisher Information Matrix and other information theoretic measures. The security of other existing schemes will be studied in the framework of chaos, to compare them to spread-spectrum and chaotic iterations. We will thus wonder whether chaotic iterations are stego-secure, or not. Last, but not least, the way to understand these chaotic tools in terms of data hiding aims will be deepened: this study is required to make chaotic-security framework really useful in practice.

\bibliographystyle{abbrv}
\bibliography{ih10}

\begin{thebibliography}{10}

\bibitem{AdelsbachKS06}
A.~Adelsbach, S.~Katzenbeisser, and A.-R. Sadeghi.
\newblock A computational model for watermark robustness.
\newblock In Camenisch et~al. \cite{DBLP:conf/ih/2006}, pages 145--160.

\bibitem{Adler65}
R.~L. Adler, A.~G. Konheim, and M.~H. McAndrew.
\newblock Topological entropy.
\newblock {\em Trans. Amer. Math. Soc.}, 114:309--319, 1965.

\bibitem{arxiv}
J.~M. Bahi and C.~Guyeux.
\newblock A watermarking algorithm satisfying topological chaos properties.
\newblock {\em CoRR}, abs/0810.4713, 2008.

\bibitem{guyeux10}
J.~M. Bahi and C.~Guyeux.
\newblock Topological chaos and chaotic iterations, application to hash
  functions.
\newblock In {\em WCCI'10: 2010 IEEE World Congress on Computational
  Intelligence}, 2010.
\newblock Accepted paper.

\bibitem{Banks92}
J.~Banks, J.~Brooks, G.~Cairns, and P.~Stacey.
\newblock On devaney's definition of chaos.
\newblock {\em Amer. Math. Monthly}, 99:332--334, 1992.

\bibitem{BarniBF03}
M.~Barni, F.~Bartolini, and T.~Furon.
\newblock A general framework for robust watermarking security.
\newblock {\em Signal Processing}, 83(10):2069--2084, 2003.
\newblock Special issue on Security of Data Hiding Technologies, invited paper.

\bibitem{DBLP:conf/iwdw/2005}
M.~Barni, I.~J. Cox, T.~Kalker, and H.~J. Kim, editors.
\newblock {\em IWDW'05: 4th International Workshop on Digital Watermarking},
  volume 3710 of {\em Lecture Notes in Computer Science}, Siena, Italy,
  September 15-17 2005. Springer.

\bibitem{BattiatoCGG99}
S.~Battiato, D.~Catalano, G.~Gallo, and R.~Gennaro.
\newblock Robust watermarking for images based on color manipulation.
\newblock In Pfitzmann \cite{DBLP:conf/ih/1999}, pages 302--317.

\bibitem{Cachin2004}
C.~Cachin.
\newblock An information-theoretic model for steganography.
\newblock {\em Information and Computation}, 192:41 -- 56, 2004.

\bibitem{DBLP:conf/ih/2006}
J.~Camenisch, C.~S. Collberg, N.~F. Johnson, and P.~Sallee, editors.
\newblock {\em IH 2006: Information Hiding, 8th International Workshop}, volume
  4437 of {\em Lecture Notes in Computer Science}, Alexandria, VA, USA, July
  2007. Springer.

\bibitem{Cayre2008}
F.~Cayre and P.~Bas.
\newblock Kerckhoffs-based embedding security classes for woa data hiding.
\newblock {\em IEEE Transactions on Information Forensics and Security},
  3(1):1--15, 2008.

\bibitem{Cayre2005}
F.~Cayre, C.~Fontaine, and T.~Furon.
\newblock Watermarking security: theory and practice.
\newblock {\em IEEE Transactions on Signal Processing}, 53(10):3976--3987,
  2005.

\bibitem{ComesanaPP05bis}
P.~Comesa{\~n}a, L.~P{\'e}rez-Freire, and F.~P{\'e}rez-Gonz{\'a}lez.
\newblock Fundamentals of data hiding security and their application to
  spread-spectrum analysis.
\newblock In {\em IH'05: Information Hiding Workshop}, pages 146--160. Lectures
  Notes in Computer Science, Springer-Verlag, 2005.

\bibitem{ComesanaPP05}
P.~Comesa{\~n}a, L.~P{\'e}rez-Freire, and F.~P{\'e}rez-Gonz{\'a}lez.
\newblock The return of the sensitivity attack.
\newblock In Barni et~al. \cite{DBLP:conf/iwdw/2005}, pages 260--274.

\bibitem{Devaney}
R.~L. Devaney.
\newblock {\em An Introduction to Chaotic Dynamical Systems, 2nd Edition}.
\newblock Westview Pr., March 2003.

\bibitem{Furon2002}
T.~Furon.
\newblock Security analysis, 2002.
\newblock European Project IST-1999-10987 CERTIMARK, Deliverable D.5.5.

\bibitem{Furon05}
T.~Furon.
\newblock A survey of watermarking security.
\newblock In Barni et~al. \cite{DBLP:conf/iwdw/2005}, pages 201--215.

\bibitem{Furon08}
T.~Furon and P.~Bas.
\newblock Broken arrows.
\newblock {\em EURASIP J. Inf. Secur.}, 2008:1--13, 2008.

\bibitem{guyeux09}
C.~Guyeux and J.~M. Bahi.
\newblock Hash functions using chaotic iterations.
\newblock {\em Journal of Algorithms \& Computational Technology}, 2009.
\newblock To appear.

\bibitem{Kalker2001}
T.~Kalker.
\newblock Considerations on watermarking security.
\newblock pages 201--206, 2001.

\bibitem{Ker06}
A.~D. Ker.
\newblock Batch steganography and pooled steganalysis.
\newblock In Camenisch et~al. \cite{DBLP:conf/ih/2006}, pages 265--281.

\bibitem{Li75}
T.~Y. Li and J.~A. Yorke.
\newblock Period three implies chaos.
\newblock {\em Amer. Math. Monthly}, 82(10):985--992, 1975.

\bibitem{Malvar03}
H.~Malvar and D.~Florêncio.
\newblock Improved spread spectrum: A new modulation technique for robust
  watermarking.
\newblock {\em IEEE Trans. Signal Proceeding}, 53:898--905, 2003.

\bibitem{Mittelholzer99}
T.~Mittelholzer.
\newblock An information-theoretic approach to steganography and watermarking.
\newblock In Pfitzmann \cite{DBLP:conf/ih/1999}, pages 1--16.

\bibitem{Perez-Freire06}
L.~Perez-Freire, P.~Comesana, J.~R. Troncoso-Pastoriza, and F.~Perez-Gonzalez.
\newblock Watermarking security: a survey.
\newblock In {\em LNCS Transactions on Data Hiding and Multimedia Security},
  2006.

\bibitem{Perez-Freire2006:Security}
L.~P{\'e}rez-Freire, F.~P{\'e}rez-González, T.~Furon, and P.~Comesaña.
\newblock Security of lattice-based data hiding against the known message
  attack.
\newblock {\em IEEE Trans. on Information Forensics and Security},
  1(4):421--439, dec 2006.

\bibitem{Perez06}
L.~Perez-Freire, F.~Pérez-gonzalez, and P.~Comesaña.
\newblock Secret dither estimation in lattice-quantization data hiding: A
  set-membership approach.
\newblock In E.~J. Delp and P.~W. Wong, editors, {\em Security, Steganography,
  and Watermarking of Multimedia Contents}, San Jose, California, USA, January
  2006. SPIE.

\bibitem{DBLP:conf/ih/1999}
A.~Pfitzmann, editor.
\newblock {\em IH'99: 3rd International Workshop on Information Hiding}, volume
  1768 of {\em Lecture Notes in Computer Science}, Dresden, Germany, September
  29 - October 1. 2000. Springer.

\bibitem{Shannon49}
C.~E. Shannon.
\newblock Communication theory of secrecy systems.
\newblock {\em Bell Systems Technical Journal}, 28:656--715, 1949.

\bibitem{Simmons83}
G.~J. Simmons.
\newblock The prisoners' problem and the subliminal channel.
\newblock In {\em Advances in Cryptology, Proc. CRYPTO'83}, pages 51--67, 1984.

\bibitem{internet09}
Q.~Wang, C.~Guyeux, and J.~M. Bahi.
\newblock A novel pseudo-random number generator based on discrete chaotic
  iterations.
\newblock In {\em INTERNET'09: First International Conference on Evolving
  Internet}, pages 71--76, 2009.

\end{thebibliography}

\end{document}